\batchmode
\documentclass[12pt,english]{article}

\usepackage[]{graphicx}\usepackage[]{color}
\makeatletter
\def\maxwidth{ %
  \ifdim\Gin@nat@width>\linewidth
    \linewidth
  \else
    \Gin@nat@width
  \fi
}
\makeatother

\definecolor{fgcolor}{rgb}{0.345, 0.345, 0.345}

\usepackage{mathtools}

\usepackage{framed}
\makeatletter
 {\par\unskip\endMakeFramed%
 \at@end@of@kframe}
\makeatother

\definecolor{shadecolor}{rgb}{.97, .97, .97}
\definecolor{messagecolor}{rgb}{0, 0, 0}
\definecolor{warningcolor}{rgb}{1, 0, 1}
\definecolor{errorcolor}{rgb}{1, 0, 0}
\usepackage{float}
\usepackage{alltt}
\usepackage{lmodern}

\usepackage[T1]{fontenc}
\usepackage[latin9]{inputenc}
\usepackage[letterpaper]{geometry}
\usepackage{color}
\usepackage{babel}
\usepackage{amsmath}
\usepackage{amsthm}
\usepackage{mathrsfs}
\usepackage{amssymb}
\usepackage{setspace}
\usepackage{tabularx}
\usepackage{graphicx}
\usepackage[unicode=true,pdfusetitle,
 bookmarks=true,bookmarksnumbered=false,bookmarksopen=false,
 breaklinks=false,pdfborder={0 0 1},backref=false,colorlinks=true]
 {hyperref}

\makeatletter
\theoremstyle{plain}
\newtheorem{assumption}{\protect\assumptionname}
\theoremstyle{plain}
\newtheorem{thm}{\protect\theoremname}
  \theoremstyle{remark}
  \newtheorem{rem}{\protect\remarkname}
  \theoremstyle{plain}
  \newtheorem{lem}{\protect\lemmaname}
  \theoremstyle{plain}
  
    \theoremstyle{plain}
    
\newtheorem{innercustomgeneric}{\customgenericname}
\providecommand{\customgenericname}{}
\newcommand{\newcustomtheorem}[2]{%
  \newenvironment{#1}[1]
  {%
   \renewcommand\customgenericname{#2}%
   \renewcommand\theinnercustomgeneric{##1}%
   \innercustomgeneric
  }
  {\endinnercustomgeneric}
}

\newcustomtheorem{customthm}{Theorem}
\newcustomtheorem{customlem}{Lemma}
\newcustomtheorem{customassumption}{Assumption}

\DeclareMathOperator{\sign}{sign}
\@ifundefined{date}{}{\date{}}
\newcommand\independent{\protect\mathpalette{\protect\independenT}{\perp}}
\def\independenT#1#2{\mathrel{\rlap{$#1#2$}\mkern2mu{#1#2}}}

\usepackage{bbm}
\usepackage{chngcntr}
\hypersetup{urlcolor=blue}

\usepackage{indentfirst}

\makeatother

  \providecommand{\lemmaname}{Lemma}
  \providecommand{\remarkname}{Remark}
\providecommand{\corollaryname}{Corollary}
\providecommand{\theoremname}{Theorem}
\providecommand{\examplename}{Example}
\providecommand{\assumptionname}{Assumption}
\newcommand{\etal}{\textit{et al}.\enspace}
\IfFileExists{upquote.sty}{\usepackage{upquote}}{}

\usepackage{natbib}
\usepackage{caption}
\begin{document}

\begin{titlepage}
\title{Robust Principal Component Analysis with Non-Sparse Errors
}
\author{Jushan Bai\thanks{Department of Economics, Columbia University; jushan.bai@columbia.edu.
}\and Junlong Feng\thanks{Department of Economics, Columbia University; junlong.feng@columbia.edu.}}
\date{\today}
\maketitle
\begin{abstract}
\noindent  We show that when a high-dimensional data matrix is the sum of a low-rank matrix and a random error matrix with independent entries, the low-rank component can be consistently estimated by solving a convex minimization problem.
We develop a new theoretical argument to establish consistency without assuming sparsity or the existence of any moments of the error matrix, so that fat-tailed continuous random errors such as Cauchy are allowed. The results are  illustrated by simulations.\\
\vspace{0in}
\vspace{0in}\\

\bigskip
\end{abstract}
\setcounter{page}{0}
\thispagestyle{empty}
\end{titlepage}

\section{Introduction}

A low-rank component in high dimensional data sets is often the object of interest. In asset return analysis, for example, a low-rank matrix represents systematic risks (Ross, 1976). In psychology, the main personality traits form a low-rank matrix (e.g., Cattell, 1978 and Goldberg, 1990). In background/foreground separation, by stacking the pixels of each frame of a video in a column vector, the static background is a rank one component in the resulting matrix because it stays unchanged across frames, see Bouwmans \etal(2017) for a survey. In gene expression prediction, the gene expression values may form a low-rank matrix because genes act in groups and at the expression levels they are interdependent (Kapur \etal, 2016).

To fix ideas, in this paper, we assume the  data matrix, $Y$, is of $N\times T$ dimension and  consists of a low-rank component $L_{0}$, namely
\begin{equation}\label{eq1}
Y=L_{0}+Z_{0}
\end{equation}
where $rank(L_{0})=r$ is small but unknown; $L_{0}$ can be random or deterministic. The magnitude of the its elements are allowed to diverge to infinity with $N$ and $T$ grow; $Z_{0}$ is a random error matrix with median zero entries that have positive densities around $0$.

We propose to estimate $L_{0}$ using a variant of the Principal Component Pursuit (PCP), introduced and studied by Cand\`es \etal(2011), Chandrasekaran \etal (2011) \textit{etc}. We show the estimator is consistent for $L_{0}$ in the Frobenius norm under certain conditions. This is the first time that consistency is established with continuous and potentially fat-tailed random errors. Formally, the estimator $\hat{L}$ is as follows
\begin{align}
\hat{L}=\arg\min_{L}\  \ ||L||_{*}+&\lambda	||Y-L||_{1},~~~~~\ s.t.\ ||L||_{\infty}\leq \alpha \label{eq2}
\end{align}
where $||\cdot||_{*}$ denotes the nuclear norm. $||\cdot||_{1}$ and $||\cdot||_{\infty}$ are the $\ell_{1}$ norm and the $\ell_{\infty}$ norm of a matrix seen as a long vector. Both $\lambda$ and $\alpha$ are $N,T$ dependent. In particular, $\alpha$ can grow with $N$ and $T$.  We call $\hat{L}$ the \textit{Bounded Principal Component Pursuit} (BPCP) as compared to the standard PCP, it has an additional constraint bounding the max entry of the estimator by $\alpha$.

As a preview of how the estimator works, first note that the nuclear norm is the convex envelope of the rank of a matrix because by definition, the nuclear norm is the $\ell_{1}$ norm of the vector of singular values while the rank is its $\ell_{0}$ norm, that is, the number of nonzero elements in the vector. Thus, minimizing the nuclear norm is a convex surrogate of rank minimization.

The other term in the objective function is the $\ell_{1}$ penalty for the residuals to induce robustness. The intuition is analogous to the LAD (least absolute deviation) estimator for linear regression; it is well-known that minimizing the sum of absolute deviations is robust to fat-tailed errors. It turns out for the BPCP estimator to work, the errors essentially only need to have zero median, just like the standard LAD  estimator.

Finally, the constraint in \eqref{eq2} is needed for technical reasons. We allow $\alpha$ to be $N,T$ dependent and can go to infinity as $N$ and $T$ increase. We will be precise about the allowed rate of such divergence. As it turns out, the restriction is actually very mild because in many stochastic models of $L_{0}$, $||L_{0}||_{\infty}$ diverges slower than the rate allowed with high probability. Therefore, imposing the constraint is without loss of generality in these models.

This paper adds to the theory of PCP and some of its variants, developed in Cand\`es \etal(2011), Chandrasekaran \etal (2011), Ganesh \etal(2010), Zhou \etal(2010), \textit{etc}.  In Cand\`es \etal(2011), it is assumed that $Z_{0}$ follows a Bernoulli model, i.e., each element in $Z_{0}$ is equal to $0$ with probability $1-\rho$. They show that when $1-\rho$ is large enough, PCP can \textit{exactly} recover both $L_{0}$ and $Z_{0}$ with high probability. The proof hinges on the existence of a matrix called \textit{dual certificate}, which relies on the sparsity of $Z_{0}$.  Ganesh \etal(2010) generalize the result by allowing for an arbitrarily small but fixed $1-\rho$ and show a dual certificate still exists. However, their results are not applicable to a continuously distributed error matrix because then for any entry $Z_{0,it}$, $P(Z_{0,it}=0)=0$ by definition, and thus $1-\rho=0$. Zhou \etal(2010) study a variant of PCP called the \textit{Stable Principal Component Pursuit} (SPCP). In their model, $Y$ is known to consist of a low rank, a sparse, and a dense component. They minimize a similar objective function over both $L$ and $Z$ with a constraint bounding the difference between the sum of them and $Y$ in the Frobenius norm. They show that the Frobenius norm of the estimation error is bounded by the Frobenius norm of the dense component multiplied by $\max\{N,T\}$. This bound is evidently too large for consistency. Hsu \etal(2011) change the objective function by adding a squared Frobenius norm penalty for the difference between $Y$ and the sum of the low-rank and the sparse components. They prove nuclear norm consistency for the low-rank matrix provided that the sparse component has an increasing fraction of zeros and the dense errors are Gaussian (see the second example in their section D). All the existing work in the above literature require $Z_{0}$ to have  a certain fraction of entries to be $0$ with positive probability. In contrast, 
in this paper all entries in $Z_{0}$ can be nonzero almost surely. Agarwal \etal(2012) study a broad class of models allowing $Y$ to be determined by a general operator of $L_{0}$ and $Z_{0}$, which are not necessarily to be exactly low-rank or sparse. Similar to Hsu \etal(2011), their objective function also has the additional Frobenius norm penalty as they allow for the existence of an additional noise component whose operator norm is not too large. Their results are more comparable with ours because they also allow for a $Z_{0}$ with all entries nonzero. However, to obtain consistency, $Z_{0}$ needs to be approximately sparse, i.e., the fraction of the entries that are large in magnitude needs to be shrinking to $0$ and the sum of the absolute values of the rest entries is $o_{p}(NT)$. This condition rejects many random models for $Z_{0}$, especially if $Z_{0}$ has very fat tails. By contrast, this paper only focuses on the linear decomposition model \eqref{eq1} and under stronger assumptions including a probabilistic model for $Z_{0}$, consistency is established even if all entries of $Z_{0}$ are nonzero and most of them are large in magnitude; we do not put any restrictions on the entries' tail distributions, so long as they have zero median and positive densities around $0$.


This paper also lies in the broader literature of estimating low-rank components in various settings. The following is only a small portion of many contributions in this literature. Tao and Yuan (2011), Xu \etal(2012), Wong and Lee (2017) and  Brahma \etal(2018) study the case where both a sparse component and a dense noise component exist besides the low-rank component in the decomposition. Xu \etal(2012) replaces the $\ell_{1}$ norm in the objective function with the $\ell_{1,2}$ norm. Wong and Lee (2017) changes it to the Huber loss function. Brahma \etal(2018) allows other general forms of penalty, such as SCAD penalty. To achieve consistency, they all need the noise matrix to have small norms. Wright \etal(2013) study the noiseless case but $Y$ is a compressive measurement of $L_{0}+Z_{0}$. Bach \etal(2008) and Negahban and Wainwright (2011) study regression models with low-rank coefficient matrices. Bai and Li (2012) impose a factor structure on the low-rank component and estimate it using MLE. Xie and Xing (2014) consider the Principal Component Analysis explicitly assuming Cauchy noise under the MLE framework and their minimization problem is nonconvex. Cai \etal(2010), Cand\`es and Recht (2009), and Cand\`es and Tao (2010), among others, study matrix completion problem for a low-rank matrix. Xia \etal(2018) develop methods for estimating low-rank tensors.

The rest of the paper is organized as follows. Section \ref{sec2} introduces the main technical tool  we developed for this problem, named \textit{Bernoulli device}. This device decomposes $Z_{0}$ into a matrix $D_{0}$ only containing $Z_{0}$'s small entries and a matrix $S_{0}$ with large entries of $Z_{0}$. Section \ref{sec3} extends results in Ganesh \etal(2010) by showing a \textit{dual certificate} exists even if the fraction of zero entries in $S_{0}$ decreases to $0$ slowly. Section \ref{sec4} presents the key condition for consistency derived from the optimality condition by using the dual certificate and exploiting the complementary structure of $D_{0}$ and $S_{0}$ endowed by the Bernoulli device. Section \ref{sec5} states and proves the main theorems of the paper. Simulation results are demonstrated in Section \ref{sec6}. Section \ref{sec7} concludes. Proofs of some of the lemmas are contained in the Appendix.

\subsection*{Notation}
Throughout, $||\cdot||_{*}$,~ $||\cdot||_{1}$, ~$||\cdot||_{F}$ and $||\cdot||_{\infty}$ denote the nuclear norm, the $\ell_{1}$ norm, the Frobenius norm, and the max norm of a matrix. $||\cdot||$ denotes the Euclidean norm of a vector, or the operator norm of a matrix or an operator. For the same matrix, $||\cdot||\leq||\cdot||_{F}\leq||\cdot||_{*}$ and $||\cdot||_{F}\leq||\cdot||_{1}$. For two generic scalars $a$ and $b$,  denote $a\land b\equiv \min\{a,b\}$ and $a\lor b\equiv \max\{a,b\}$. For any positive sequences $a$ and $b$, $a\asymp b$ means there exist $0<c_{1}\leq c_{2}<\infty$ such that $c_{1}a\leq b\leq c_{2}a$. For any matrices $A$ and $B$ of the same size, $A\circ B$ is the componentwise product of A and B. For any two random objects $X$ and $Y$, denote independence between $X$ and $Y$ by $X\independent Y$. Finally, $C$, $C'$, $C_{1}$ and $C_{2}$ denote generic positive constants that may be different in different uses.

\section{A Bernoulli Device}\label{sec2}
The object of interest in this paper is
\begin{equation}\label{eq3}
\frac{1}{NT}||\hat{L}-L_{0}||_{F}^{2}
\end{equation}
where $\hat{L}$ is defined in \eqref{eq2}. To bound this quantity and obtain consistency, we follow the idea in Cand\`es \etal(2011), Ganesh \etal(2010) and Zhou \etal(2010) to use a \textit{dual certificate}, a matrix which will be defined in the next section, to derive the optimality condition for \eqref{eq2}. This condition will then yield a bound for \eqref{eq3}. The main theoretical challenge is that in the first place, the existence of a dual certificate hinges on the existence of zero entries (with positive probability) in the error matrix, a luxury we do not have in this paper. The key idea is to decompose $Z_{0}$ into $S_{0}+D_{0}$ in such a way that i) a large enough fraction of entries in $S_{0}$ are $0$ with positive probability to guarantee the existence of a dual certificate, and ii) that fraction cannot be too large, on the other hand, so that \eqref{eq3} can be bounded by a function of $||D_{0}||_{1}$ that converges to $0$ in probability. We begin with constructing this decomposition using a Bernoulli device.

We first introduce the following assumption.
\begin{assumption}\label{ass1}
a) $L_{0}\independent  Z_{0}$.
b) $Z_{0,it}$ are independent and $med(Z_{0,it})=0$;
c) The set of densities of $Z_{0,it}$, $\{f_{Z_{0,it}}\}_{N,T}$ are equicontinuous and uniformly bounded away from $0$ at $0$.
\end{assumption}
Note that equicontinuity in c) can be replaced with continuity if we strengthen the independence condition in part b) to be i.i.d.

Under Assumption \ref{ass1}, let $\{\delta\}$ be a positive sequence such that $\delta\to 0$ as $N,T\to\infty$.
Let $(\underline{\gamma}_{it},\bar{\gamma}_{it})$ be a pair of constants satisfying
\begin{equation}\label{eq4}
\begin{aligned}
P(Z_{0,it}\geq \bar{\gamma}_{it})&=P(Z_{0,it}\leq \underline{\gamma}_{it})=\frac{1-\delta}{2}
\end{aligned}
\end{equation}
Assumption \ref{ass1} b) and c) guarantees the existence and uniqueness of such a pair for large enough $N,T$ while $\delta$ approaches $0$.

With $\underline{\gamma}_{it}$ and $\bar{\gamma}_{it}$, let $M$ be an $N\times T$ matrix whose entries are defined by
\begin{equation}\label{eq5}
M_{it}=\mathbbm{1}(\underline{\gamma}_{it}<Z_{0,it}<\bar{\gamma}_{it})
\end{equation}
where $\mathbbm{1}(\cdot)$ is the indicator function. Then let $D_{0}=M\circ Z_{0}$ and $S_{0}=Z_{0}-M\circ Z_{0}$ and we have
\begin{equation}\label{eq6}
Z_{0}=S_{0}+D_{0}
\end{equation}
Under Assumption \ref{ass1}, $D_{0}$ and $S_{0}$ have the following properties:
\begin{enumerate}
\item Both $S_{0}$ and $D_{0}$ contain $0$ entries and their locations are complementary;
\item By construction, $P(S_{0,it}=0)=\delta$ for any $(i,t)$. Meanwhile, by Assumption \ref{ass1} b) and c), for small enough $\delta$, $P(Z_{0,it}\leq 0)-P(Z_{0,it}\leq \underline{\gamma}_{it})=\frac{\delta}{2}=-f(\tilde{\gamma}_{it})\underline{\gamma}_{it}$ by the Mean Value Theorem where $\tilde{\gamma}_{it}$ lies between $\underline{\gamma}_{it}$ and $0$. Since $f_{Z_{0,it}}(0)$ is uniformly bounded away from $0$ over $N,T$ for small enough $\delta$, there exists uniform constants $C>0$ such that $-\underline{\gamma}_{it}<C\delta$. Similar results hold for $\bar{\gamma}_{it}$. Therefore, there exist $C'>0$ such that $|D_{0,it}|<C'\delta$ uniformly.

\item Let $E=\sign(S_{0})$, the the entries in $E$ are i.i.d. with $P(E_{it}=0)=\delta$ and $P(E_{it}=1)=P(Z_{0,it}\geq \bar{\gamma}_{it})=P(Z_{0,it}\leq \underline{\gamma}_{it})=P(E_{it}=-1)=\frac{1-\delta}{2}$.
\end{enumerate}

The Bernoulii device $M$ thus delivers a pair $(D_{0},S_{0})$ that achieves the two goals described in the beginning of this section. First, as will be seen, items 2 and 3 guarantee the existence of a dual certificate under certain conditions. Second, by Hoeffding's inequality, the order of $||D_{0}||_{1}$ is no greater than $C'NT\delta^{2}$ in high probability provided that $\delta$ converges to $0$ at an appropriate rate. Note this holds regardless of how the distribution of the original error $Z_{0,it}$ behaves except the requirements for the zero-median and the positive and continuous density at $0$.
\section{Dual Certificate}\label{sec3}

In this section, we treat $S_{0}$ as the ``sparse'' error matrix and show that although $\delta$ decreases to $0$, a dual certificate that is similar to Ganesh \etal(2010) exists. As mentioned in the introduction, Ganesh \etal(2010) show a dual certificate exists for any small yet fixed $\delta$. We extend their results by carefully choosing the rate of $\delta$, $\lambda$ and other constants in their proof. We closely follow their construction of the dual certificate but for completeness, we record it here and shall indicate where the construction needs to be modified to handle a shrinking $\delta$ by construction.

First we need an identification condition to guarantee $L_{0}$ to be non-sparse so that it is distinguishable from $S_{0}$. We adopt the incoherence condition in Cand\`es and Recht (2009), Cand\`es \etal(2011), Ganesh \etal(2010), \textit{etc}. Besides it, as there is an additional constraint in \eqref{eq2}, we need a condition to guarantee $L_{0}$ to be a feasible solution. Let $U\Sigma V^{*}$ be a singular value decomposition of $L_{0}$, where $U$ and $V$ are $N\times r$ and $T\times	r$ matrices of left and right singular vectors and $\Sigma$ is an $r\times r$ diagonal matrix with singular values in descending order on its diagonal.
\begin{assumption}\label{ass2}
There exists a constant $C$ such that with probability approaching 1,
\begin{align}
\max_{i}||U^{*}e_{i}||^{2}\leq C\frac{\mu r}{N},&\ \max_{t}||V^{*}e_{t}||^{2}\leq C \frac{\mu r}{T}\label{eq8}\\
||L_{0}||_{\infty}&\leq \alpha\label{eq9}
\end{align}
where $\mu$, $r$ and $\alpha$ can be $N,T$ dependent. $(e_{i})_{i}$ and $(e_{t})_{t}$ are canonical bases of N- and T- dimensional linear spaces.
\end{assumption}
Inequality \eqref{eq8} in Assumption \ref{ass2} is the incoherence condition, stating that the singular vectors of $L_{0}$ are well-spread. A direct and useful consequence of \eqref{eq8} is that
\begin{equation}\label{eq10}
||UV^{*}||_{\infty}\leq C \frac{\mu r}{\sqrt{NT}}
\end{equation}
 by noticing that $||UV^{*}||_{\infty}=\max_{it}|\sum_{k=1}^{r}U_{ik}V_{tk}|\leq\sqrt{\sum_{k=1}^{r}|U_{ik}|^{2}}\cdot \sqrt{\sum_{k=1}^{r}|V_{tk}|^{2}}\leq C\frac{\mu r}{\sqrt{NT}}$ where the first inequality follows from the Cauchy-Schwarz inequality and the second is from \eqref{eq8}. Here $\mu$ characterizes how coherent the singular vectors are with the canonical bases. It can be $N,T$ dependent and is allowed to diverge to $\infty$. Cand\`es and Recht (2009) provide examples where $\mu=O(\log(N\lor T))$ and one of them is the \textit{random orthogonal model} in which the columns in $U$ and $V$ are sampled uniformly among all families of $r$ orthonormal vectors independently of each other. Fan \etal(2018) also give an example where $\mu=O\big(\sqrt{\log(N\lor T)}\big)$.

Inequality \eqref{eq9} is an inclusion assumption which implies $L_{0}$ is a feasible solution with probability approaching 1. It restricts the magnitude of the maximal entry in $L_{0}$. Again $\alpha$ is allowed to increase to $\infty$ with $N$ and $T$. Note that \eqref{eq9}  and \eqref{eq10} imply that $L_{0}$'s largest singular value $\sigma_{1}\leq \frac{\alpha}{C}\frac{\sqrt{NT}}{\mu r}$ because $L_{0}=\sum_{k=1}^{r}u_{k}v^{*}_{k}\sigma_{k}$ while $UV^{*}=\sum_{k=1}^{r}u_{k}v^{*}_{k}$, where $u_{k}$ and $v_{k}$ are the $k$th column of $U$ and $V$, respectively.

Before we define the dual certificate, it is useful to introduce some notations.

Let $\Phi$ be the linear space of matrices
\begin{equation*}
\Phi\equiv \{UX^{*}+YV^{*},X\in \mathbb{R}^{T\times r}, Y\in\mathbb{R}^{N\times r}\}
\end{equation*}
and let its orthogonal complement be $\Phi^{\perp}$. Denote the linear projection onto $\Phi$ and $\Phi^{\perp}$ by $\mathcal{P}_{\Phi}$ and $\mathcal{P}_{\Phi^{\perp}}$, respectively. Then it can be shown that for any $N\times T$ matrix $R$ (e.g. Cand\`es and Recht (2009)),
\begin{equation*}
\mathcal{P}_{\Phi}R=UU^{*}R+RVV^{*}-UU^{*}RVV^{*}
\end{equation*}
and
\begin{equation*}
\mathcal{P}_{\Phi^{\perp}}R=(I-UU^{*})R(I-VV^{*})
\end{equation*}

Let $\Omega$ be the support of $S_{0}$, i.e., $\Omega=\{(i,t):S_{0,it}\neq 0\}$. With a slight abuse of notation, we also denote the linear space of matrices supported on $\Omega$ by $\Omega$. The projection onto this space is denoted by $\mathcal{P}_{\Omega}$. Specifically, for an $N\times T$ matrix $R$,
\begin{equation}\label{eq11}
(\mathcal{P}_{\Omega}R)_{i,t}=\mathbbm{1}\big((i,t)\in\Omega\big)\cdot R_{it}
\end{equation}
The complement of the support set $\Omega$ is denoted by $\Omega^{c}$. Let the linear space of matrices supported on it be $\Omega^{\perp}$ and the projection onto the space be $\mathcal{P}_{\Omega^{\perp}}$, defined similarly as \eqref{eq11}.

Finally, we characterize the subgradient of $||L||_{*}$ and $||S||_{1}$ using these notations. The subgradient of $||L||_{*}$ evaluated at $L_{0}$ is equal to $UV^{*}+W$ where $\mathcal{P}_{\Phi}W=0$ and $||W||\leq 1$. Meanwhile, recall that $E$ denotes the sign of $S_{0}$, so the subgradient of $||S||_{1}$ at $S_{0}$ is equal to $E+F$, where $\mathcal{P}_{\Omega}F=0$ and $||F||_{\infty}\leq 1$.

Now we are ready to define the dual certificate $W$ as any $N\times T$ matrix satisfying the following conditions:
\begin{equation}\label{eq12}
\left\{ \begin{array}{c}
\mathcal{P}_{\Phi}W=0\\
||W||\leq \frac{1}{2}\\
||\mathcal{P}_{\Omega}(UV^{*}+W-\lambda E)||_{F}\leq \frac{\lambda\delta}{16} \\
||\mathcal{P}_{\Omega^{\perp}}(UV^{*}+W)||_{\infty}<\frac{\lambda}{2}
\end{array}\right.
\end{equation}

Note this definition is very similar to equation (6) in Ganesh \etal(2010). The only important difference occurs on the right hand side of the third inequality.

Now we present the construction of $W$ that is similar to Ganesh \etal(2010) with necessary modifications to accommodate $\delta\to 0$.

Let $W=W_{L}+W_{S}$:
\begin{itemize}
\item Construction of $W_{L}$. As we have a more delicate random model and more structures regarding $D_{0}$ and $S_{0}$, we need more subtle argument to justify the construction of $W_{L}$. For any $(i,t)$, write $D_{0,it}$ and $S_{0,it}$ as follows:
\begin{align*}
D_{0,it}=&\Delta_{1,it}\cdot\tilde{Z}_{1,it}\\
S_{0,it}=&(1-\Delta_{1,it})\cdot\big(\Delta_{2,it}\cdot \tilde{Z}_{2,it}+(1-\Delta_{2,it})\cdot \tilde{Z}_{3,it}\big)\\
Z_{0,it}=&D_{0,it}+S_{0,it}
\end{align*}
where $\Delta_{1,it}=1-\prod_{j=1}^{j_{0}}(1-\tilde{\Delta}_{j,it})$, $\tilde{\Delta}_{j,it}\stackrel{i.i.d.}{\sim}Ber(q)$ such that $1-\delta=(1-q)^{j_{0}}$. $\Delta_{2,it}\stackrel{i.i.d.}{\sim}Ber(1/2)$. $\tilde{Z}_{1,it}$, $\tilde{Z}_{2,it}$ and $\tilde{Z}_{3,it}$ follow $Z_{0,it}$'s distribution truncated between $\underline{\gamma}_{it}$ and $\bar{\gamma}_{it}$, below $\underline{\gamma}_{it}$, and above $\bar{\gamma}_{it}$, respectively. All these Bernoulli random variables and $\tilde{Z}_{1,it}$, $\tilde{Z}_{2,it}$ and $\tilde{Z}_{3,it}$ are independent. It can be verified the distribution of $(D_{0,it},S_{0,it})$ as well as $Z_{0,it}$ are the same as the original model, so the two models are equivalent.

Now let $\Omega_{j}=\{(i,t):\tilde{\Delta}_{j,it}=1\}$. By construction, $\Omega^{c}=\cup_{j=1}^{j_{0}}\Omega_{j}$. Unlike Ganesh \etal(2010) in which $j_{0}=2\log(N)$, in this paper we need $j_{0}$ to be finite such that $q$ converges to $0$ at the same rate as $\delta$. As shown in Lemma \ref{lem2}, $j_{0}=4$ is sufficient. Now let $Q_{0}=0$, and
\begin{equation}\label{13}
Q_{j}=Q_{j-1}+q^{-1}\mathcal{P}_{\Omega_{j}}\mathcal{P}_{\Phi}(UV^{*}-Q_{j-1}),j=1,2,...,j_{0}
\end{equation}
Finally, let $W_{L}=\mathcal{P}_{\Phi^{\perp}}Q_{j_{0}}$. This construction is called the \textit{golfing scheme}; it was first developed by Gross (2011) and Gross \etal(2010) for matrix completion problems and was later extended to matrix separation problems by Cand\`es \etal(2011).
\item Construction of $W_{S}$. Let $W_{S}=\lambda \mathcal{P}_{\Phi^{\perp}}\sum_{k\geq 0}(\mathcal{P}_{\Omega}\mathcal{P}_{\Phi}\mathcal{P}_{\Omega})^{k}E$, provided that $||\mathcal{P}_{\Omega}\mathcal{P}_{\Phi}\mathcal{P}_{\Omega}||=||\mathcal{P}_{\Omega}\mathcal{P}_{\Phi}||^{2}<1$ to guarantee the Neumann series $\sum_{k\geq 0}(\mathcal{P}_{\Omega}\mathcal{P}_{\Phi}\mathcal{P}_{\Omega})^{k}$ is well defined and is equal to $(\mathcal{P}_{\Omega}-\mathcal{P}_{\Omega}\mathcal{P}_{\Phi}\mathcal{P}_{\Omega})^{-1}$. As pointed out by Ganesh \etal(2010), $W_{S}$ can be viewed as constructed using least squares; it has the smallest Frobenius norm among matrices $\tilde{W}$ satisfying $\mathcal{P}_{\Omega}\tilde{W}=\lambda E$ and $\mathcal{P}_{\Phi}\tilde{W}=0$.
\end{itemize}

The following lemmas then provide sufficient conditions for $W=W_{L}+W_{S}$ to satisfy \eqref{eq12}. The proof of Lemma \ref{lem1} is omitted because it is immediately implied by Theorem 2.6 in Cand\`es \etal(2011), stated in the Appendix. The proof of Lemma \ref{lem2} follow Lemma 3 and 4 in Ganesh \etal(2010) closely, but are tailored in a way to allow $\delta\to 0$. For completeness, they are contained in the Appendix.

\begin{lem}\label{lem1}
Suppose $\delta\geq C\frac{\mu r\log (N\lor T)}{\varepsilon^{2}(N\land T)}$. Then under Assumptions \ref{ass1} and \ref{ass2}, $||\mathcal{P}_{\Omega}\mathcal{P}_{\Phi}||^{2}\leq 1-\delta+\varepsilon\delta$ with high probability.
\end{lem}
\begin{lem}\label{lem2}
Suppose $N\asymp T$, $\mu>\log N$, $\frac{\mu^{11/3}r^{3}}{N^{1/3}}=o(1)$, $\delta\asymp \frac{\mu r}{N^{1/3}}$, $\lambda\asymp \frac{\mu^{1/3}}{N^{2/3}}$, $\varepsilon\asymp\frac{\log N}{N^{1/3}}$. If $j_{0}=4$, then under Assumptions \ref{ass1} and \ref{ass2}, with high probability
\newline
a). $||W_{L}||\leq 1/4$,
\newline
b). $||\mathcal{P}_{\Omega}(UV^{*}+W_{L})||_{F}<\frac{\lambda\delta}{16}$,
\newline
c). $||\mathcal{P}_{\Omega^{\perp}}(UV^{*}+W_{L})||_{\infty}<\frac{\lambda}{4}$.
\newline
d). $||W_{S}||<1/4$,
\newline
e). $||\mathcal{P}_{\Omega^{\perp}}W_{S}||_{\infty}<\frac{\lambda}{4}$.
\end{lem}
\begin{rem}
The condition $N\asymp T$ requires $N$ and $T$ to diverge to infinity at the same rate. Divergence of both $N$ and $T$ is a theoretical feature when modeling the high-dimensional data, yet the rate condition is indeed stronger than in the existing PCP literature as most of the work do not restrict the rate at all. Relaxation can be made to some extent by tuning the rate of other parameters like $\delta$ and $\lambda$.
\end{rem}
\begin{rem}
The rates for $\delta$, $\lambda$ and $\varepsilon$ in Lemma \ref{lem2} are not unique. As will be seen in the next section, the convergence rate of $\frac{1}{NT}||\hat{L}-L_{0}||_{F}^{2}$ will be determined only by $\delta$, $\alpha$, $r$ and $\mu$, and the choice of the rates in Lemma \ref{lem1} yields the fastest converging $\delta$ while $W$ satisfies \eqref{eq12}.
\end{rem}
\begin{rem}
The rank is allowed to diverge as long as $r\leq C N^{1/9}\mu^{-11/9}(\log N)^{-1}$. Note this rate is smaller than the rate allowed in Ganesh \etal(2010), which is $CN\mu^{-1}(\log N)^{-2}$. The loss is unavoidable because from the condition in Lemma \ref{lem1}, instead of being constant, $\delta$ now decreases to $0$ so the order allowed for $r$ is smaller. The insight is that we are trading off between sparsity of the error matrix and sparsity of the vector of the singular values of $L_{0}$. It suggests when the random error matrix is continuous, the estimator may not perform well in finite sample for $L_{0}$ with relatively high rank.
\end{rem}

\section{Optimality Condition}\label{sec4}
In this section, using the dual certificate $W$ and the complementarity of the support sets of $S_{0}$ and $D_{0}$, we derive the optimality condition for the minimization problem in \eqref{eq2} and as it turns out, the condition induces an upper bound for \eqref{eq3} that facilitates the analysis of consistency.

For any feasible solution $L$ to \eqref{eq2}, let $Z\equiv Y-L$. Let $H\equiv L-L_{0}$ so $Z=Z_{0}-H$. By the Bernoulli device, we have the following properties: i) if we consider the difference $H_{S}\equiv Z-S_{0}$, then by construction $H_{S}=D_{0}-H$, and ii) $\mathcal{P}_{\Omega^{\perp}}D_{0}=D_{0}$. Utilizing them, we have the following lemma.
\begin{lem}\label{lem4}
Suppose $||\mathcal{P}_{\Omega}\mathcal{P}_{\Phi}||^{2}\leq 1-\delta+\varepsilon\delta$, $\delta\to 0$, $\varepsilon\to 0$ as $N,T\to\infty$, and the dual certificate $W$ satisfying \eqref{eq12} exists. For any disturbance $(H,-H)$ at $(L_{0},Z_{0})$, if
\begin{equation}\label{eq13}
||L_{0}+H||_{*}+\lambda ||Z_{0}-H||_{1}\leq ||L_{0}||_{*}+\lambda ||Z_{0}||_{1},
\end{equation}
then for large enough $N$ and $T$,
\begin{equation}\label{eq14}
\frac{1}{\lambda}||\mathcal{P}_{\Phi^{\perp}}H||_{*}+ ||\mathcal{P}_{\Omega^{\perp}}H||_{1}\leq 8||D_{0}||_{1}
\end{equation}
\end{lem}
\begin{proof}
Since $H_{S}+S_{0}=Z=Z_{0}-H$, $||L_{0}+H||_{*}+\lambda ||Z_{0}-H||_{1}=||L_{0}+H||_{*}+\lambda ||S_{0}+H_{S}||_{1}$. Let $X_{L}$ and $X_{S}$ be the subgradients of $||\cdot||_{*}$ and $||\cdot||_{1}$ at $L_{0}$ and $S_{0}$. We have the identities: $X_{L}=UV^{*}+W+\mathcal{P}_{\Phi^{\perp}}(X_{L}-UV^{*}-W)$, and $\lambda X_{S}=UV^{*}+W-\mathcal{P}_{\Omega}(UV^{*}+W-\lambda	E)+\mathcal{P}_{\Omega^{\perp}}(\lambda X_{S}-UV^{*}-W)$.

By the definition of subgradient, we have
\begin{align}
||L_{0}+H||_{*}+\lambda ||S_{0}+H_{S}||_{1}\geq& ||L_{0}||_{*}+\lambda ||S_{0}||_{1}+\langle X_{L},H\rangle+\lambda\langle X_{S},H_{S}\rangle\notag\\
=& ||L_{0}||_{*}+\lambda ||S_{0}||_{1}+\langle UV^{*}+W,H\rangle+\langle\mathcal{P}_{\Phi^{\perp}}(X_{L}-UV^{*}-W),H\rangle\notag\\
&+\langle UV^{*}+W-\mathcal{P}_{\Omega}(UV^{*}+W-\lambda	E),H_{S}\rangle\notag\\
&+\langle \mathcal{P}_{\Omega^{\perp}}(\lambda X_{S}-UV^{*}-W),H_{S}\rangle\notag\\
=&||L_{0}||_{*}+\lambda ||S_{0}||_{1}+\langle UV^{*}+W,D_{0}\rangle+\langle X_{L}-UV^{*}-W,\mathcal{P}_{\Phi^{\perp}}(H)\rangle\notag\\
&-\langle \mathcal{P}_{\Omega}(UV^{*}+W-\lambda	E),H_{S}\rangle+\langle \lambda X_{S}-UV^{*}-W,\mathcal{P}_{\Omega^{\perp}}(H_{S})\rangle\notag\\
\geq&||L_{0}||_{*}+\lambda ||S_{0}||_{1}+\langle UV^{*}+W,D_{0}\rangle+\langle X_{L}-UV^{*}-W,\mathcal{P}_{\Phi^{\perp}}(H)\rangle\notag\\
&+\langle \lambda X_{S}-UV^{*}-W,\mathcal{P}_{\Omega^{\perp}}(H_{S})\rangle-\frac{1}{16}\lambda \delta 	||\mathcal{P}_{\Omega} H_{S}||_{F}\label{14}
\end{align}
where the first inequality follows from the definition of subgradient. The following equality is obtained by substituting the two identities of $X_{L}$ and $X_{S}$ into the right hand side. The second equality uses the relationship $H=D_{0}-H_{S}$, as well as self-adjointness of $\mathcal{P}_{\Omega}$. The last inequality follows from \eqref{eq12}.

We now manipulate the last four terms on the right hand side of the last inequality.

For the first term, since $D_{0}=\mathcal{P}_{\Omega^{\perp}}D_{0}$ by construction,
\begin{equation}\label{15}
\langle UV^{*}+W,D_{0}\rangle=\langle UV^{*}+W,\mathcal{P}_{\Omega^{\perp}}D_{0}\rangle
\geq -||\mathcal{P}_{\Omega^{\perp}}(UV^{*}+W)||_{\infty}||D_{0}||_{1}
\geq -\frac{\lambda}{2}||D_{0}||_{1}
\end{equation}

For the second term, by duality, let $X_{L}$ be such that $\langle X_{L}, \mathcal{P}_{\Phi^{\perp}}(H)\rangle =||\mathcal{P}_{\Phi^{\perp}}(H)||_{*}$. Also, note that $|\langle UV^{*}+W,\mathcal{P}_{\Phi^{\perp}}(H)\rangle|\leq ||\mathcal{P}_{\Phi^{\perp}}(UV^{*}+W)||\cdot||\mathcal{P}_{\Phi^{\perp}}(H)||_{*}\leq \frac{1}{2}||\mathcal{P}_{\Phi^{\perp}}(H)||_{*}$ since $\mathcal{P}_{\Phi^{\perp}}(UV^{*})=0$ and $\mathcal{P}_{\Phi^{\perp}}W=W$. Therefore, we have
\begin{equation}\label{16}
\langle X_{L}-UV^{*}-W,\mathcal{P}_{\Phi^{\perp}}(H)\rangle\geq \frac{1}{2}||\mathcal{P}_{\Phi^{\perp}}(H)||_{*}
\end{equation}

For the third term, let $X_{S}=\sign(\mathcal{P}_{\Omega^{\perp}}H_{S})$ so $\langle X_{S}, \mathcal{P}_{\Omega^{\perp}}(H_{S})\rangle=||\mathcal{P}_{\Omega^{\perp}}(H_{S})||_{1}$. Then
\begin{equation}\label{17}
\langle \lambda X_{S}-UV^{*}-W,\mathcal{P}_{\Omega^{\perp}}(H_{S})\rangle\geq \big(\lambda-||\mathcal{P}_{\Omega^{\perp}}(UV^{*}+W)||_{\infty}\big)||\mathcal{P}_{\Omega^{\perp}}(H_{S})||_{1}\geq \frac{\lambda}{2}||\mathcal{P}_{\Omega^{\perp}}(H_{S})||_{1}
\end{equation}

Finally for the last term, since $H_{S}=D_{0}-H$ and $\mathcal{P}_{\Omega}D_{0}=0$, we have $||\mathcal{P}_{\Omega}H_{S}||_{F}=||\mathcal{P}_{\Omega}H||_{F}$ and
\begin{align*}
||\mathcal{P}_{\Omega}H||_{F}\leq& ||\mathcal{P}_{\Omega}\mathcal{P}_{\Phi}H||_{F}+||\mathcal{P}_{\Omega}\mathcal{P}_{\Phi^{\perp}}H||_{F}\\
\leq & \sqrt{1-\delta+\varepsilon \delta}||H||_{F}+||\mathcal{P}_{\Phi^{\perp}}H||_{F}\\
\leq&  \sqrt{1-\delta+\varepsilon \delta}(||\mathcal{P}_{\Omega}H||_{F}+||\mathcal{P}_{\Omega^{\perp}}H||_{F})+||\mathcal{P}_{\Phi^{\perp}}H||_{F}\\
\leq & \sqrt{1-\delta+\varepsilon \delta}(||\mathcal{P}_{\Omega}H||_{F}+||\mathcal{P}_{\Omega^{\perp}}H_{S}||_{1}+||\mathcal{P}_{\Omega^{\perp}}D_{0}||_{1})+||\mathcal{P}_{\Phi^{\perp}}H||_{*}
\end{align*}
where the last inequality follows from the fact that $||\cdot||_{F}\leq ||\cdot||_{1}$ and $||\cdot||_{F}\leq ||\cdot||_{*}$.
By rearranging the terms and using $D_{0}=\mathcal{P}_{\Omega^{\perp}}D_{0}$ again,
\begin{equation}\label{eq18}
||\mathcal{P}_{\Omega}H||_{F}\leq \frac{ \sqrt{1-\delta+\varepsilon\delta}}{1- \sqrt{1-\delta+\varepsilon\delta}}(||\mathcal{P}_{\Omega^{\perp}}H_{S}||_{1}+||D_{0}||_{1})+\frac{1}{1- \sqrt{1-\delta+\varepsilon\delta}}||\mathcal{P}_{\Phi^{\perp}}H||_{*}
\end{equation}

Combining the last inequality in \eqref{14} with \eqref{15}, \eqref{16}, \eqref{17} and \eqref{eq18}, we have
\begin{align*}
||L_{0}+H||_{*}+\lambda||Z_{0}-H||_{1}\geq& ||L_{0}||_{*}+\lambda||S_{0}||_{1}+(\frac{1}{2}-\frac{\lambda\delta}{16(1-\sqrt{1-\delta+\varepsilon\delta})})||\mathcal{P}_{\Phi^{\perp}}H||_{*}\\
&+\lambda (\frac{1}{2}-\frac{ \delta\sqrt{1-\delta+\varepsilon\delta}}{16(1- \sqrt{1-\delta+\varepsilon\delta})})||\mathcal{P}_{\Omega^{\perp}}H_{S}||_{1}\\
&-\frac{\lambda }{2}||D_{0}||_{1}-\frac{ \lambda\delta\sqrt{1-\delta+\varepsilon\delta}}{16(1- \sqrt{1-\delta+\varepsilon\delta})} ||D_{0}||_{1}\\
\geq &||L_{0}||_{*}+\lambda||S_{0}||_{1}+\frac{1}{4}||\mathcal{P}_{\Phi^{\perp}}H||_{*}+\frac{\lambda}{4}||\mathcal{P}_{\Omega^{\perp}}H_{S}||_{1}-\frac{3}{4}\lambda ||D_{0}||_{1}
\end{align*}
where the second inequality holds because
\begin{equation}\label{19}
\frac{\delta}{1- \sqrt{1-\delta+\varepsilon\delta}}=\frac{\delta(1+ \sqrt{1-\delta+\varepsilon\delta})}{\delta(1-\varepsilon)}\leq 4
\end{equation}
and
\begin{equation}\label{eq19}
\frac{\delta\sqrt{1-\delta+\varepsilon\delta}}{1- \sqrt{1-\delta+\varepsilon\delta}}=\frac{\sqrt{1-\delta+\varepsilon\delta}(1+ \sqrt{1-\delta+\varepsilon\delta})}{1-\varepsilon}\leq 4
\end{equation}
for large enough $N$.

Then again, since $S_{0}=Z_{0}-D_{0}$ and $\mathcal{P}_{\Omega^{\perp}}H_{S}=\mathcal{P}_{\Omega^{\perp}}(D_{0}-H)=D_{0}-\mathcal{P}_{\Omega^{\perp}}H$,
\begin{equation*}
\begin{aligned}
&||L_{0}||_{*}+\lambda||S_{0}||_{1}+\frac{1}{4}||\mathcal{P}_{\Phi^{\perp}}H||_{*}+\frac{\lambda}{4}||\mathcal{P}_{\Omega^{\perp}}H_{S}||_{1}-\frac{3}{4}\lambda ||D_{0}||_{1}\\
\geq &||L_{0}||_{*}+\lambda||Z_{0}||_{1}+\frac{1}{4}||\mathcal{P}_{\Phi^{\perp}}H||_{*}+\frac{\lambda}{4}||\mathcal{P}_{\Omega^{\perp}}H||_{1}-2\lambda ||D_{0}||_{1}
\end{aligned}
\end{equation*}

Therefore, by \eqref{eq13},
\begin{align*}
||\mathcal{P}_{\Phi^{\perp}}H||_{*}+\lambda ||\mathcal{P}_{\Omega^{\perp}}H||_{1}\leq& 8\lambda ||D_{0}||_{1}\\
\implies \frac{1}{\lambda}||\mathcal{P}_{\Phi^{\perp}}H||_{*}+ ||\mathcal{P}_{\Omega^{\perp}}H||_{1}\leq&  8||D_{0}||_{1}
\end{align*}
\end{proof}
\begin{rem}\label{rem4}
Under the event that $||L_{0}||_{\infty}\leq \alpha$, $L_{0}$ is a feasible solution to the minimization problem in \eqref{eq2}, so $\hat{L}$ yields a smaller objective function than $L_{0}$ by definition. Then equation \eqref{eq13} is satisfied so the bound given by Lemma \ref{lem4} applies to the estimation error $\hat{L}-L_{0}$.
\end{rem}
\begin{rem}
The basic idea of the proof is in line with Lemma 2.5 in Cand\`es \etal(2011) and Lemma 5 in Zhou \etal(2010). However there are major differences which lead to a tighter bound than the one obtained in Zhou \etal(2010) for SPCP. By construction, the argument in $||\cdot||_{*}$ and $||\cdot||_{1}$ in our objective function in \eqref{eq2} always add up to $Y$, while when we take the subderivative for $||\cdot||_{1}$, we take it at $S_{0}$ instead of $Z_{0}$. This enables us to utilize both the dual certificate as well as the nice relationship between $H_{S}$ and $H$ and the relationship between the support sets of $D_{0}$ and $S_{0}$, all generated by the Bernoulli device. These special properties in turn yield a tighter bound in terms of $||D_{0}||_{1}$.
\end{rem}

\section{Main Results}\label{sec5}
From Lemma \ref{lem4}, we obtain bounds for $||\mathcal{P}_{\Phi^{\perp}}H||_{*}$ and $||\mathcal{P}_{\Omega^{\perp}}H||_{1}$. The next lemma bounds the corresponding norms of their complements.
\begin{lem}\label{lem5}
Suppose $||\mathcal{P}_{\Omega}\mathcal{P}_{\Phi}||^{2}\leq 1-\delta+\varepsilon\delta$, $\delta\to 0$ and $\varepsilon\to 0$ as $N,T\to\infty$, then for large enough $N$ and $T$,
\begin{equation}\label{eq21}
||\mathcal{P}_{\Phi}H-\mathcal{P}_{\Omega}H||_{F}^{2}\geq \frac{\delta}{4}(||\mathcal{P}_{\Phi}H||_{F}^{2}+||\mathcal{P}_{\Omega}H||_{F}^{2})
\end{equation}
\end{lem}
\begin{proof}
\begin{align*}
||\mathcal{P}_{\Phi}H-\mathcal{P}_{\Omega}H||_{F}^{2}=&||\mathcal{P}_{\Phi}H||_{F}^{2}+||\mathcal{P}_{\Omega}H||_{F}^{2}-2\langle \mathcal{P}_{\Phi}H, \mathcal{P}_{\Omega}H\rangle\\
\geq & ||\mathcal{P}_{\Phi}H||_{F}^{2}+||\mathcal{P}_{\Omega}H||_{F}^{2}-2\sqrt{1-\delta+\varepsilon\delta}||\mathcal{P}_{\Phi}H||_{F}||\mathcal{P}_{\Omega}H||_{F}\\
\geq & (1-\sqrt{1-\delta+\varepsilon\delta})(||\mathcal{P}_{\Phi}H||_{F}^{2}+||\mathcal{P}_{\Omega}H||_{F}^{2})\\
= &\frac{\delta(1-\varepsilon)}{1+\sqrt{1-\delta+\varepsilon\delta}}(||\mathcal{P}_{\Phi}H||_{F}^{2}+||\mathcal{P}_{\Omega}H||_{F}^{2})
\end{align*}
where the first inequality follows from $|\langle \mathcal{P}_{\Phi}H, \mathcal{P}_{\Omega}H\rangle|=|\langle \mathcal{P}_{\Omega}H, \mathcal{P}_{\Omega}\mathcal{P}_{\Phi}\mathcal{P}_{\Phi}H\rangle|\leq ||\mathcal{P}_{\Omega}\mathcal{P}_{\Phi}||\cdot || \mathcal{P}_{\Omega}H||_{F}\cdot  ||\mathcal{P}_{\Phi}H||_{F}$.
Then for large enough $N$ and $T$, $1-\varepsilon>\frac{1}{2}$ and $1+\sqrt{1-\delta+\varepsilon\delta}<2$. This completes the proof.
\end{proof}
Given Lemmas \ref{lem4} and \ref{lem5}, now we are ready to state and prove the main result in the paper.
\begin{thm} \label{thm1}
Under Assumptions \ref{ass1}, \ref{ass2} and the conditions in Lemma \ref{lem2}, if $\alpha \mu r=o(N^{1/3})$, there exists $C>0$ such that
\begin{equation}\label{eq22}
P\big(\frac{1}{NT}||\hat{L}-L_{0}||_{F}^{2}\leq C\big(\alpha\lor (\mu^{8/3}r^{2})\big)\delta\big)\to 1
\end{equation}
where $\big(\alpha\lor (\mu^{8/3}r^{2})\big)\delta=\frac{(\alpha\mu r)\lor (\mu^{11/3}r^{3}) }{N^{1/3}}=o(1)$.
\end{thm}
\begin{proof}
Under the event
\begin{equation*}
\mathcal{E}\equiv\{||\mathcal{P}_{\Omega}\mathcal{P}_{\Phi}||^{2}\leq 1-\delta+\varepsilon\delta\}\cap \{W\ \textrm{satisfies}\ \eqref{eq12}\} \cap \{\textrm{conditions}\ \eqref{eq8}\ and\ \eqref{eq9}\ \textrm{holds}\},
\end{equation*}
let $H=\hat{L}-L_{0}$ and by Lemma \ref{lem4},
\begin{equation}\label{25}
||\mathcal{P}_{\Phi^{\perp}}H||_{F}^{2}\leq 64\lambda^{2}||D_{0}||_{1}^{2}
\end{equation}
\begin{equation}\label{26}
||\mathcal{P}_{\Omega^{\perp}}H||_{F}^{2}\leq 16\alpha||D_{0}||_{1}
\end{equation}
where \eqref{25} follows  from $||\mathcal{P}_{\Phi^{\perp}}H||_{F}\leq ||\mathcal{P}_{\Phi^{\perp}}H||_{*}$ and \eqref{26} follows from $||\mathcal{P}_{\Omega^{\perp}}H||_{F}^{2}\leq ||\mathcal{P}_{\Omega^{\perp}}H||_{\infty}\cdot ||\mathcal{P}_{\Omega^{\perp}}H||_{1}$. Notice under $\mathcal{E}$, $||L_{0}||_{\infty}\leq \alpha$, so $||\mathcal{P}_{\Omega^\perp}H||_{\infty}\leq ||H||_{\infty}\leq ||L_{0}||_{\infty}+||\hat{L}||_{\infty}\leq 2\alpha$.

Since $H-H=\mathcal{P}_{\Phi}H+\mathcal{P}_{\Phi^{\perp}}H-\mathcal{P}_{\Omega}H-\mathcal{P}_{\Omega^{\perp}}H=0$, we have $||\mathcal{P}_{\Phi}H-\mathcal{P}_{\Omega}H||_{F}=||\mathcal{P}_{\Phi^{\perp}}H-\mathcal{P}_{\Omega^{\perp}}H||_{F}$. By Lemma \ref{lem5},
\begin{align*}
\frac{\delta}{4}(||\mathcal{P}_{\Phi}H||_{F}^{2}+||\mathcal{P}_{\Omega}H||_{F}^{2})\leq& ||\mathcal{P}_{\Phi}H-\mathcal{P}_{\Omega}H||_{F}^{2}\\
=&||\mathcal{P}_{\Phi^{\perp}}H-\mathcal{P}_{\Omega^{\perp}}H||_{F}^{2}\\
\leq & 4\big(||\mathcal{P}_{\Phi^{\perp}}H||_{F}^{2}\lor ||\mathcal{P}_{\Omega^{\perp}}H||_{F}^{2}\big)
\end{align*}
Therefore,
\begin{align}
||H||_{F}^{2}=&\frac{1}{2}\big(||\mathcal{P}_{\Omega^{\perp}}H||_{F}^{2}+||\mathcal{P}_{\Omega}H||_{F}^{2}+||\mathcal{P}_{\Phi^{\perp}}H||_{F}^{2}+||\mathcal{P}_{\Phi}H||_{F}^{2}\big)\\
\leq &\frac{C_{1}}{\delta}(\alpha||D_{0}||_{1}\lor \lambda^{2}||D_{0}||_{1}^{2})\label{28}
\end{align}

Next we derive the order of $||D_{0}||_{1}$. Note $|D_{0,it}|\leq C'\delta$ by construction. Therefore,
\begin{equation*}
P(||D_{0}||_{1}-C'NT\delta^{2} \geq C't\delta)\leq P(\delta\sum_{it}M_{it}-NT\delta^{2}\geq t\delta)
=  P(\sum_{it}M_{it}-NT\delta\geq t)
\leq  \exp\big(-\frac{2t^{2}}{NT}\big)
\end{equation*}
where $M_{it}$ is defined in \eqref{eq5} and the last inequality follows from Hoeffding's inequality. Let $t=N\log N$, then $\exp(-\frac{2t^{2}}{NT})\to 0$. Meanwhile, since $NT\delta\asymp N^{5/3}\mu r$, it dominates $t$. Therefore, $||D_{0}||_{1}\leq C_{2}NT\delta^{2}$ with probability approaching 1.

Now conditional on the event $\bar{\mathcal{E}}\equiv\mathcal{E}\cap \{||D_{0}||_{1}\leq C_{2}NT\delta^{2}\}$,  \eqref{28} implies
\begin{align*}
||H||_{F}^{2}\leq& \frac{C}{\delta}\big((\alpha NT\delta^{2})\lor (\lambda^{2}(NT)^{2}\delta^{4})\big)\\
\leq & C\delta \big((\alpha NT)\lor (\lambda^{2}(NT)^{2}\delta^{2})\big)\\
\implies \frac{1}{NT}||H||_{F}^{2}\leq &C\big(\alpha\lor (\mu^{8/3}r^{2})\big)\delta
\end{align*}
This completes the proof as $\bar{\mathcal{E}}$ occurs with probability approaching 1.
\end{proof}

\begin{rem}
Theorem \ref{thm1} indicates that the rate of convergence is determined by $\delta$, $r$, $\mu$ and $\alpha$. As the driving force of consistency, the faster $\delta$ converges to $0$, the faster the convergence rate of $\frac{1}{NT}||\hat{L}-L_{0}||_{F}^{2}$ will be. However, the rate of $\delta$ is bounded because the fraction of $0$ entries in $S_{0}$ cannot be too small for the existence of a dual certificate. The parameters $r$, $\mu$ and $\alpha$ slow the convergence down, but if we look at the relative bias, $\frac{||\hat{L}-L_{0}||_{F}^{2}}{||L_{0}||_{F}^{2}}$, the effects of these parameters may be mitigated or reversed. For example suppose $L_{0}$ follows the random orthogonal model introduced in Section 3 and all the $r$ singular values are of same order. Then $\alpha\asymp \frac{\sigma_{1}\mu r}{\sqrt{NT}}$ where $\sigma_{1}$ is the largest singular value, so the error bound is proportional to $(\sigma_{1} \mu^{2} r^{2})\lor (\mu^{11/3}r^{3})$. Meanwhile, the order of $||L_{0}||_{F}^{2}$ is proportional to $\sigma_{1}^{2}\mu^{2}r^{2}$. Therefore, the upper bound of the relative bias is proportional to $\frac{1}{\sigma_{1}}\lor\frac{\mu^{5/3}r}{\sigma^{2}_{1}}$, and we can see the effect of $\mu$ and $r$ are smaller and the larger the singular values are, which, for fixed $\mu$ and $r$, implies a larger $\alpha$, the smaller the relative bias is. This is because $L_{0}$ with larger singular values tend to dominate $Z_{0}$ in the decomposition, and thus recovering it from the observed data matrix $Y$ is relatively easier.
\end{rem}

\section{Simulations}\label{sec6}
In this section, we present two simulation experiments to illustrate the effectiveness of the BPCP estimator. In the first experiment, we generate random $L_{0}$ and $Z_{0}$ with $r=1,3,5$. In each case, we examine the performance of the estimator for both Gaussian and Cauchy error matrices by comparing $\frac{1}{NT}||\hat{L}-L_{0}||_{F}^{2}$ as well as the relative estimation error $||\hat{L}-L_{0}||_{F}^{2}/||L_{0}||_{F}^{2}$ in each case as we gradually increase the sample size. In the second experiment, we fix a picture as $L_{0}$, and superimpose it with Gaussian and Cauchy white noise. The purpose of this example is to visually show how the estimator performs when the error is continuously fat-tailed distributed.

To implement the estimator, note the minimization problem in \eqref{eq2} is equivalent as
\begin{equation}
\min_{L,Z}\ \ ||L||_{*}+\lambda||Z||_{1},\ s.t.\ L+Z=Y,\ \ ||L||_{\infty}\leq \alpha
\end{equation}
We first set $\alpha$ to be a large number and solve the problem without the inequality constraint, then verify whether the solution satisfies the inequality. For the first step, we adopt the Augmented Lagrangian Multiplier algorithm (ALM) studied in Lin \etal(2010) and Yuan and Yang (2013). The algorithm solves the following problem
\begin{equation}
\min_{L,Z}\ \ ||L||_{*}+\lambda||Z||_{1}+\langle \Lambda, Y-L-Z\rangle+\frac{\nu}{2}||Y-L-Z||_{F}^{2}
\end{equation}
where $\Lambda$ is the Lagrangian multiplier for the equality constraint and the last term is a penalty for deviating from the constraint. The algorithm solves the minimization problem iteratively and terminates if both $||Y-L_{k}-Z_{k}||_{F}$ and $\nu ||Z_{k}-Z_{k-1}||_{F}$ are small enough, where the subscript $k$ denotes the $k$th iteration. We set the stopping criteria to be $10^{-7}||Y||_{F}$ and $10^{-5}$ respectively. Following Yuan and Yang (2013) and Cand\`es \etal(2011), we set $\nu=\frac{NT}{4||Y||_{1}}$. Finally, $\lambda$ is set to be $0.7(\frac{\log(N\land T)}{NT})^{1/3}$ in the first experiment and $0.5(\frac{\log(N\land T)}{NT})^{1/3}$ in the second; both of them satisfy the conditions in the theory when $\mu$ is of the order $\log(N)$ and $N\asymp T$.

\subsection{Numerical Experiment}
In this experiment, we set $N=T=200,300,400,500$ and $r=1,3,5$. In each case, we draw an $N\times r$ and an $r\times T$ matrix from $N(0,1)$ and their product is $L_{0}$. For Gaussian error, entries of $Z_{0}$ are independently drawn from $N(0,1)$ while for Cauchy error, they are drawn from the standard Cauchy distribution.

Table \ref{tab1} shows the value $\frac{1}{NT}||\hat{L}-L_{0}||_{F}^{2}$ in each case. It decreases as $N$ and $T$ increase. Also, consistent with Theorem \ref{thm1}, fixing $N$ and $T$, the estimation error increases as the rank increases. However, the magnitude of the increase is smaller than indicated by Theorem \ref{thm1}, where the upper bound is proportional to $r^{3}$. This suggests our bound may be further improved and deriving the lower bound may also be useful. Another observation is that the estimation error under Cauchy errors is systematically bigger than under Gaussian errors. To explain this phenomenon, note that the standard Cauchy density is lower than the standard normal density around $0$, so for any given $\delta$ that is small enough, $-\underline{\gamma}_{it}$ and $\bar{\gamma}_{it}$ introduced in Section \ref{sec2} are bigger for the Cauchy error, resulting in bigger $||D_{0}||_{1}$. This difference is negligible when $N,T\to\infty$ but leads to finite sample differences in the error bound.
\begin{table}[H]
\centering
\caption{Consistency of $\frac{1}{NT}||\hat{L}-L_{0}||_{F}^{2}$}\label{tab1}
\begin{tabular}{m{3cm} m{1.6cm} m{1.6cm} m{1.6cm} m{1.6cm} m{1.6cm} m{1.6cm}}
\hline
\hline
 & \multicolumn{3}{c}{Gaussian Error}  & \multicolumn{3}{c}{Cauchy Error} \\
 \cline{2-7}
 & $r=1$ & $r=3$ & $r=5$ & $r=1$ & $r=3$ & $r=5$\\
 \hline
$N=T=200$ & $0.0461$ &$0.1349$& $0.2535$ & $0.0812$ & $0.3002$ &$0.5491$\\
$N=T=300$ & $0.0322$& $0.0943$&$0.1692$ &$0.0533$ &$0.1999$&$0.3676$\\
$N=T=400$ & $0.0247$ & $0. 0736$&$0.1308$ & $0.0438$ &$0.1496$&$0.2764$\\
$N=T=500$ & $0.0209$ &$0.0619$ &$0.1078$ &$0.0359$ &$0.1198$& $0.2151$\\
\hline
    \end{tabular}
    \end{table}

Table \ref{tab2} shows the relative estimation error $\frac{||\hat{L}-L_{0}||_{F}^{2}}{||L||_{F}^{2}}$. We can see that in every column, this quantity is also decreasing. Under the same error distribution, the differences across ranks are now significantly smaller. This is because as the rank increases, $||L_{0}||_{F}^{2}$ also increases by construction. Specifically, since $L_{0}$ is the product of two independent matrices with i.i.d. standard normal entries, $||L_{0}||_{F}^{2}$ is of the order of $NTr$, so dividing it mitigates the effect of $r$.

\begin{table}[H]
\centering
\caption{Relative Error $\frac{||\hat{L}-L_{0}||_{F}^{2}}{||L_{0}||_{F}^{2}}$}\label{tab2}
\begin{tabular}{m{3cm} m{1.6cm} m{1.6cm} m{1.6cm} m{1.6cm} m{1.6cm} m{1.6cm}}
\hline
\hline
 & \multicolumn{3}{c}{Gaussian Error}  & \multicolumn{3}{c}{Cauchy Error} \\
 \cline{2-7}
 & $r=1$ & $r=3$ & $r=5$ & $r=1$ & $r=3$ & $r=5$\\
 \hline
$N=T=200$ & $0.0378$ &$0.0411$& $0.0482$ & $0.0679$ & $0.1082$ &$0.1139$\\
$N=T=300$ & $0.0265$& $0.0303$&$0.0334$ &$0.0459$ &$0.0687$&$0.0729$\\
$N=T=400$ & $0.0224$ & $0. 0243$&$0.0265$ & $0.0383$ &$0.0498$&$0.0538$\\
$N=T=500$ & $0.0197$ &$0.0203$ &$0.0220$ &$0.0310$ &$0.0395$& $0.0416$\\
\hline
    \end{tabular}
    \end{table}

 \subsection{Graphical Experiment}

 In this experiment, we superimpose white noises drawn from the standard normal or the standard Cauchy distribution on a picture. The picture we use is directly downloaded from a built-in example in MATLAB (the file name is eight.tif). The picture has resolution $242\times 308$. We stack all columns into a long vector and duplicate it for 199 times, obtaining a $74536\times 200$ matrix $L_{0}$. By construction, $L_{0}$ has rank $1$ because all columns are equal. Then we draw a $74536\times 200$ matrix with i.i.d. entries from either of the two distributions as $Z_{0}$ and add it to $L_{0}$. This procedure simulates 200 frames from a video of the static picture interfered by white noise.

Figure \ref{fig1} and Figure \ref{fig2} display the results. In each figure, the northwest (NW) is the original picture. There are four coins in it, two heads and two tails. The southwest (SW) shows one of the 200 frames after the the picture is superimposed with Gaussian noises (Figure \ref{fig1}) or Cauchy noises (Figure \ref{fig2}). It can be seen that the details of the coins are no longer unrecognizable. The northeast (SE) quadrant shows the recovered picture. Differences in this picture between Figure \ref{fig1} and Figure \ref{fig2} are hardly to be seen, except the background in Figure \ref{fig2} is slightly darker. The southeast quadrant shows the same frame of residuals, in which we cannot see contours of the coins, indicating it contains very few information about the original picture.
 \begin{figure}[H]
\centering
\includegraphics[width=.45\textwidth]{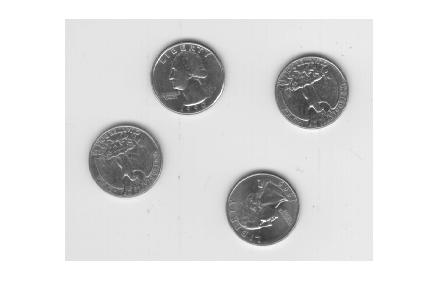}\quad
\includegraphics[width=.45\textwidth]{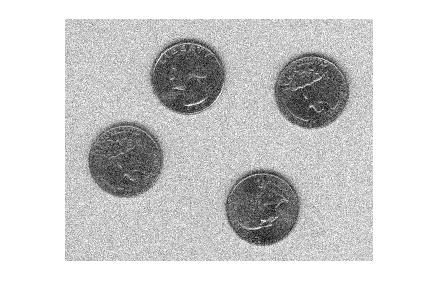}


\includegraphics[width=.45\textwidth]{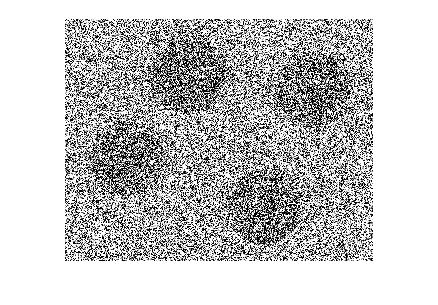}\quad
\includegraphics[width=.45\textwidth]{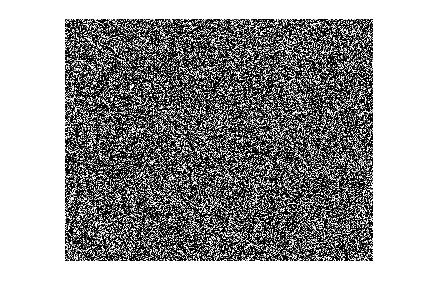}
\caption{Original (NW), Gaussian Noise (SW), Recovered Picture (NE) and Residual (SE)}
\label{fig1}
\end{figure}

 \begin{figure}[h]
\centering
\includegraphics[width=.45\textwidth]{Original.jpg}\quad
\includegraphics[width=.45\textwidth]{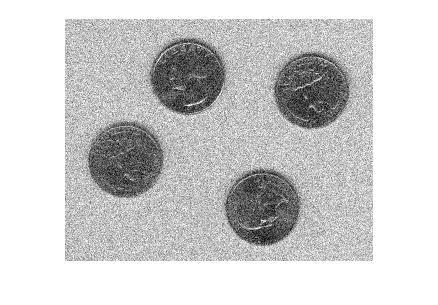}


\includegraphics[width=.45\textwidth]{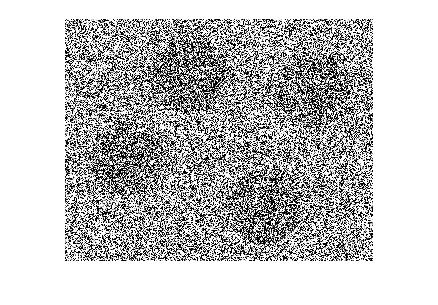}\quad
\includegraphics[width=.45\textwidth]{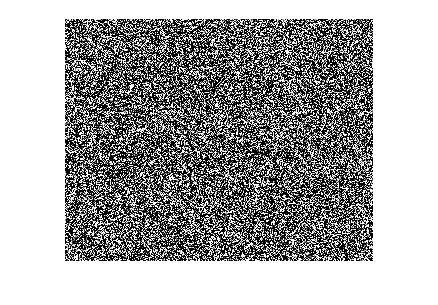}
\caption{Original (NW), Cauchy Noise (SW), Recovered Picture (NE) and Residual (SE)}
\label{fig2}
\end{figure}
\section{Conclusion}\label{sec7}
This paper provides sufficient conditions under which the Bounded Principal Component Pursuit (BPCP) is consistent for the random or deterministic low-rank component, whose entries can go to infinity with $N$ and $T$. Built on the prior work on Principal Component Pursuit (PCP), consistency is shown by constructing a Bernoulli device which induces sufficiently tight bounds for the estimation error whose order does not depend on the moments of the error $Z_{0}$; the random error only needs to have zero median, analogous to the LAD estimator for linear models. The results indicate the estimator is highly robust to large outliers.

There are a few aspects of the estimator that remains to be studied. This paper only shows consistency in the Frobenius norm. It would be interesting to see whether consistency holds componentwisely or in the $\ell_{\infty}$ norm. These aspects will be studied in future work.
\appendix
\section*{Appendix}
\subsection*{Proof of Lemma \ref{lem2}}
The proof of Lemma \ref{lem2} closely follows the proof of Lemmas 3 and 4 in Ganesh \etal(2010) and the proof of Lemmas 2.8 and 2.9 in Cand\`es \etal(2011). We record it here only to make it easy to see where modifications are made to accommodate a converging $\delta$ by construction. Throughout, we condition on the event that condition \eqref{eq8} in Assumption \ref{ass2} holds.

We begin by stating three lemmas in  Cand\`es and Recht (2009) and Cand\`es \etal(2011).
\begin{customlem}{A1}[Cand\`es and Recht (2009), Theoerm 4.1; Cand\`es \etal(2011), Theorem 2.6]\label{lemA3}
Suppose $\Omega_{0}\sim Ber(\delta_{0})$. Then with high probability,
\begin{equation*}
||\mathcal{P}_{\Phi}-\delta_{0}^{-1}\mathcal{P}_{\Phi}\mathcal{P}_{\Omega_{0}}\mathcal{P}_{\Phi}||\leq \varepsilon
\end{equation*}
provided that $\delta_{0}\geq C\varepsilon^{-2}(\mu r \log (N\lor T)/(N\land T)$.
\end{customlem}
\begin{customlem}{A2}[Cand\`es \etal(2011), Lemma 3.1]\label{lemA1}
 Suppose $X\in\Phi$ is a fixed matrix, and $\Omega_{0}\sim Ber(\delta_{0})$. Then, with high probability,
\begin{equation*}
||X-\delta_{0}^{-1}\mathcal{P}_{\Phi}\mathcal{P}_{\Omega_{0}}X||_{\infty}\leq \varepsilon ||X||_{\infty}
\end{equation*}
provided that $\delta_{0}\geq C\varepsilon^{-2}(\mu r \log (N\lor T)/(N\land T)$.
\end{customlem}
\begin{customlem}{A3}[Cand\`es and Recht (2009), Theorem 6.3; Cand\`es \etal(2011), Lemma 3.2]\label{lemA2}
 Suppose $X$ is fixed, and $\Omega_{0}\sim Ber(\delta_{0})$. Then, with high probability,
\begin{equation*}
||(I-\delta_{0}^{-1}\mathcal{P}_{\Omega_{0}})X||\leq C\sqrt{\frac{(N\lor T)\log(N\lor T)}{\delta_{0}}} ||X||_{\infty}
\end{equation*}
for some small numerical constant $C>0$ provided that $\delta_{0}\geq C(\mu r \log (N\lor T)/(N\land T)$.
\end{customlem}

Under the conditions in Lemma \ref{lem2}, $N\asymp T$, so in the subsequent proof, $N\lor T$ and $N\land T$ will simply be written as $N$. Following Cand\`es \etal(2011), let $X_{j}=UV^{*}-\mathcal{P}_{\Phi}Q_{j}$ where $Q_{j}$ is defined in \eqref{13}. Then although $U$ and $V$ can be random, the three lemmas can be first proved conditional on $U$ and $V$ and since the bounds for $U$ and $V$ in Assumption \ref{ass2} are uniform, the conclusions hold unconditionally.

\subsubsection*{Proof of a)}
From \eqref{13} and the expression of $X_{j}$, $Q_{j}=Q_{j-1}+q^{-1}\mathcal{P}_{\Omega_{j}}X_{j-1},\forall j\geq 1$. Therefore,
 $Q_{j_{0}}=\sum_{j=1}q^{-1}\mathcal{P}_{\Omega_{j}}X_{j-1}$ with $Q_{0}=0$. Then
\begin{align*}
||W_{L}||=||\mathcal{P}_{\Phi^{\perp}}Q_{j_{0}}||&\leq \sum_{j}||q^{-1}\mathcal{P}_{\Phi^{\perp}}\mathcal{P}_{\Omega_{j}}X_{j-1}||\\
&=\sum_{j=1}||\mathcal{P}_{\Phi^{\perp}}(q^{-1}\mathcal{P}_{\Omega_{j}}X_{j-1}-X_{j-1})||\\
&\leq \sum_{j=1}||q^{-1}\mathcal{P}_{\Omega_{j}}X_{j-1}-X_{j-1}||\\
&\leq C\sqrt{\frac{N\log N}{q}} \sum_{j=1}||X_{j-1}||_{\infty}\\
&\leq C'\sqrt{\frac{N\log N}{q}} ||UV^{*}||_{\infty}
\end{align*}
The second line follows from $\mathcal{P}_{\Phi^{\perp}}X_{j-1}=0$. The second inequality follows from $||\mathcal{P}_{\Phi^{\perp}}||\leq 1$. The third inequality follows from Lemma \ref{lemA2} by setting $\delta_{0}=q$. Similarly, the last inequality holds by noticing $||X_{j}||_{\infty}=||(\mathcal{P}_{\Phi}-q^{-1}\mathcal{P}_{\Phi}\mathcal{P}_{\Omega_{j}}\mathcal{P}_{\Phi})X_{j-1}||_{\infty}\leq ||X_{j-1}||_{\infty}$ by Lemma \ref{lemA1}, and thus $||X_{j}||_{\infty}\leq \varepsilon^{j}||UV^{*}||_{\infty}$. To see $q$ indeed satisfies the conditions in Lemma \ref{lemA1} and Lemma \ref{lemA2}, note that since $j_{0}=4$, $(1-q)^{4}=1-\delta$, so
\begin{equation*}
q=\frac{\delta}{(1+(1-\delta)^{\frac{1}{2}})(1+(1-\delta)^{\frac{1}{4}})}\in [ \frac{\delta}{4},\delta]
\end{equation*}
Therefore, $q$ and $\delta$ are of the same order, and thus under the conditions in Lemma \ref{lem2}, the conditions in Lemma \ref{lemA1} and Lemma \ref{lemA2} are satisfied. Then the right-hand side of the last inequality is of the order $\frac{\log N \sqrt{\mu r}}{N^{1/3}}\to 0$.

\subsubsection*{Proof of b)}
The right-hand side $\frac{\lambda\delta}{16}$ is of the order of $\frac{\mu^{4/3}r}{N}$.

For the left-hand side, since $\Omega^{c}=\cup_{j=1}^{j_{0}}\Omega_{j}$, $||\mathcal{P}_{\Omega}Q_{j_{0}}||=0$. Then
\begin{equation}\tag{A1}\label{eqA1}
\mathcal{P}_{\Omega}(UV^{*}+W_{L})=\mathcal{P}_{\Omega}(UV^{*}+\mathcal{P}_{\Phi^{\perp}}Q_{j_{0}})=\mathcal{P}_{\Omega}(UV^{*}-\mathcal{P}_{\Phi}Q_{j_{0}})=\mathcal{P}_{\Omega}(X_{j_{0}})
\end{equation}
Recall $X_{j}=(\mathcal{P}_{\Phi}-q^{-1}\mathcal{P}_{\Phi}\mathcal{P}_{\Omega_{j}}\mathcal{P}_{\Phi})X_{j-1}$. So by Lemma \ref{lemA3}, $||X_{j}||_{F}\leq \varepsilon||X_{j-1}||_{F}$, where $\varepsilon$ is the same as the one in the condition in Lemma \ref{lem2} since the orders of $q$ and $\delta$ are equal. Therefore, $||X_{j_{0}}||_{F}\leq \varepsilon^{j_{0}}||UV^{*}||_{F}=\varepsilon^{4}\sqrt{r}=O(\frac{(\log N)^{4}r^{1/2}}{N^{4/3}})$ which is smaller than $\frac{\mu^{4/3}r}{N}$ with large enough $N$.

\subsubsection*{Proof of c)}
Since $UV^{*}+W_{L}=X_{j_{0}}+Q_{j_{0}}$ and $Q_{j_{0}}$ is supported on $\Omega^{c}$,
\begin{align*}
||\mathcal{P}_{\Omega^{\perp}}(UV^{*}+W_{L})||_{\infty}=&||\mathcal{P}_{\Omega^{\perp}}(X_{j_{0}}+Q_{j_{0}})||_{\infty}\\
=&||\mathcal{P}_{\Omega^{\perp}}X_{j_{0}}+Q_{j_{0}}||_{\infty}\\
\leq& ||X_{j_{0}}||_{F}+||Q_{j_{0}}||_{\infty}
\end{align*}
From b) we already have $||X_{j_{0}}||_{F}\leq \lambda/8$ for large enough $N$ because $\lambda\asymp \frac{\mu^{1/3}}{N^{2/3}}$, and thus,
\begin{align*}
||Q_{j_{0}}||_{\infty}&= ||\sum_{j=1} q^{-1}\mathcal{P}_{\Omega_{j}}X_{j-1}||_{\infty}\\
&\leq q^{-1}\sum_{j=1}||X_{j-1}||_{\infty}\\
&\leq Cq^{-1}||UV||_{\infty}\\
&=O(\frac{1}{N^{2/3}})
\end{align*}
where the second inequality follows from $||X_{j}||_{\infty}=\varepsilon^{j}||UV^{*}||_{\infty}$ which we derived in the proof of a). It is clear that the right hand side of the inequality is smaller than $\lambda/8$ for large enough $N$.

\subsubsection*{Proof of d)}
 By construction of the Bernoulli device, for any element $E_{it}$, $P(E_{it}=0)=\delta$ and $P(E_{it}=1)=P(E_{it}=-1)=\frac{1-\delta}{2}$. Also, conditional on $\Omega$, the signs of $E$ are i.i.d. symmetric. By the definition of $W_{S}$,
\begin{equation*}
W_{S}=\lambda\mathcal{P}_{\Phi ^{\perp}}E+\lambda\mathcal{P}_{\Phi^{\perp}}\sum_{k\geq 1}(\mathcal{P}_{\Omega}\mathcal{P}_{\Phi}\mathcal{P}_{\Omega})^{k}E
\end{equation*}

For the first term, $||\lambda\mathcal{P}_{\Phi ^{\perp}}E||\leq \lambda||E||\leq \lambda C\sqrt{N}$ with large probability because $E$ has i.i.d. and mean $0$ entries (see Cand\`es \etal(2011) and the reference therein).

For the second term, denote $\mathcal{R}=\mathcal{P}_{\Phi^{\perp}}\sum_{k\geq 1}(\mathcal{P}_{\Omega}\mathcal{P}_{\Phi}\mathcal{P}_{\Omega})^{k}$. Then under the event $\{||\mathcal{P}_{\Omega}\mathcal{P}_{\Phi}||^{2}\leq 1-\delta+\varepsilon\delta\}\cap\{||\mathcal{P}_{\Omega^{\perp}}\mathcal{P}_{\Phi}||^{2}\leq \delta+\varepsilon(1-\delta)\}$ which occurs with high probability under Lemma \ref{lem1}, we have
\begin{align*}
||R||&=||\mathcal{P}_{\Phi^{\perp}}\sum_{k\geq 1}(\mathcal{P}_{\Omega}\mathcal{P}_{\Phi}\mathcal{P}_{\Omega})^{k}||\\
&\leq ||\mathcal{P}_{\Phi^{\perp}}\mathcal{P}_{\Omega}\mathcal{P}_{\Phi}\mathcal{P}_{\Omega}||\cdot \sum_{k\geq 0}||(\mathcal{P}_{\Omega}\mathcal{P}_{\Phi}\mathcal{P}_{\Omega})^{k}||\\
&\leq ||\mathcal{P}_{\Phi^{\perp}}\mathcal{P}_{\Omega}\mathcal{P}_{\Phi}||\cdot ||\mathcal{P}_{\Phi}\mathcal{P}_{\Omega}||\cdot \sum_{k\geq 0}||(\mathcal{P}_{\Omega}\mathcal{P}_{\Phi})||^{2k}\\
&=||\mathcal{P}_{\Phi^{\perp}}\mathcal{P}_{\Omega^{\perp}}\mathcal{P}_{\Phi}||\cdot ||\mathcal{P}_{\Phi}\mathcal{P}_{\Omega}||\cdot \sum_{k\geq 0}||(\mathcal{P}_{\Omega}\mathcal{P}_{\Phi})||^{2k}\\
&\leq \frac{\sqrt{(\delta+\varepsilon(1-\delta))(1-\delta+\varepsilon\delta)}}{\delta(1-\varepsilon)}\\
&\leq \frac{2}{\sqrt{\delta}}
\end{align*}
where the last inequality is from the fact that for large enough $N$, $\varepsilon<\delta$.

For any $\tau\in (0,1)$, let $N_{\tau}$ and $T_{\tau}$ denote an $\tau$-net for $\mathbb{S}^{N-1}$ and $\mathbb{S}^{T-1}$, which are $N-1$ and $T-1$ dimensional unit sphere respectively. The sizes of the spheres are at most $(3/\tau)^{N}$ and $(3/\tau)^{T}$ (see Ganesh \etal(2010)). Then
\begin{equation*}
||\mathcal{R}E||=\sup_{x\in\mathbb{S}^{N-1},y\in\mathbb{S}^{T-1}}\langle y,(\mathcal{R}E)'x\rangle\leq (1-\tau)^{-2}\sup_{x\in N_{\tau},y\in T_{\tau}}\langle y,(\mathcal{R}E)'x\rangle
\end{equation*}
For a fixed pair $(x,y)\in N_{\tau}\times T_{\tau}$, define $X(x,y)\equiv \langle y,(\mathcal{R}E)'x\rangle=\langle \mathcal{R}xy',E\rangle$. Then conditional on $\Omega$, $U$ and $V$, by Hoeffding's inequality,
\begin{align*}
P\big(|X(x,y)|>t\big|\Omega, U,V\big)\leq & 2\exp\big(-\frac{2t^{2}}{||\mathcal{R}(xy')||_{F}^{2}}\big)\\
\leq & 2\exp\big(-\frac{2t^{2}}{||\mathcal{R}||_{F}^{2}}\big)\\
\leq & 2\exp\big(-\frac{\delta t^{2}}{2}\big)
\end{align*}
Then using the union bound,
\begin{align*}
P\big(\sup_{x\in N_{\tau},y\in T_{\tau}}|X(x,y)|>t\big|\Omega, U,V\big)\leq & 2(\frac{3}{\tau})^{N+T}\exp(-\frac{\delta t^{2}}{2})
\end{align*}
Therefore,
\begin{align*}
P\big(||\mathcal{R}E||>C\sqrt{\frac{N}{\delta}}\big|\Omega, U,V\big)\leq & 2(\frac{3}{\tau})^{N+T}\exp(-\frac{\delta C^{2}(1-\tau)^{4}N}{2\delta})\\
=& 2(\frac{3}{\tau})^{N+T}\exp(-\frac{ C^{2}(1-\tau)^{4}N}{2})
\end{align*}
Since $N\asymp T$, for any $\tau$, there exists a finite $C$ such that $\big(\frac{3}{\tau}\big)^{2}<\exp(\frac{C^{2}(1-\tau)^{4}}{2})$. Also as the probability bound is not a function of $\Omega$, $U$ or $V$, the inequality holds unconditionally. Hence with high probability,
\begin{equation*}
||W_{S}||\leq C\lambda\sqrt{N}(1+\frac{1}{\sqrt{\delta}})\leq C'\mu^{-1/6}r^{-1/2}\to 0
\end{equation*}

\subsubsection*{Proof of e)}
By construction $\mathcal{P}_{\Omega^{\perp}}E=0$ and $\mathcal{P}_{\Omega^{\perp}}(\mathcal{P}_{\Omega}\mathcal{P}_{\Phi}\mathcal{P}_{\Omega})^{k}E=0$ for all $k$, so we have
\begin{align*}
\mathcal{P}_{\Omega^{\perp}}W_{S}=&\lambda\mathcal{P}_{\Omega^{\perp}}\mathcal{P}_{\Phi^{\perp}}\sum_{k\geq 0}(\mathcal{P}_{\Omega}\mathcal{P}_{\Phi}\mathcal{P}_{\Omega})^{k}E\\
=&-\lambda\mathcal{P}_{\Omega^{\perp}}\mathcal{P}_{\Phi}\sum_{k\geq 0}(\mathcal{P}_{\Omega}\mathcal{P}_{\Phi}\mathcal{P}_{\Omega})^{k}E\\
=&-\lambda\mathcal{P}_{\Omega^{\perp}}\mathcal{P}_{\Phi}(\mathcal{P}_{\Omega}-\mathcal{P}_{\Omega}\mathcal{P}_{\Phi}\mathcal{P}_{\Omega})^{-1}E
\end{align*}
Therefore, $||\mathcal{P}_{\Omega^{\perp}}W_{S}||_{\infty}\leq \lambda||\mathcal{P}_{\Phi}(\mathcal{P}_{\Omega}-\mathcal{P}_{\Omega}\mathcal{P}_{\Phi}\mathcal{P}_{\Omega})^{-1}E||_{\infty}$. Let $\tilde{W}_{S}\equiv \mathcal{P}_{\Phi}(\mathcal{P}_{\Omega}-\mathcal{P}_{\Omega}\mathcal{P}_{\Phi}\mathcal{P}_{\Omega})^{-1}E $. Then
$\tilde{W}_{S,it}=\langle e_{i}e_{t}', \tilde{W}_{S}\rangle$. Let $X(i,t)=(\mathcal{P}_{\Omega}-\mathcal{P}_{\Omega}\mathcal{P}_{\Phi}\mathcal{P}_{\Omega})^{-1}\mathcal{P}_{\Omega}\mathcal{P}_{\Phi}(e_{i}e_{t}')$, then $\langle e_{i}e_{t}', \tilde{W}_{S}\rangle=\langle X(i,t), E\rangle$ by noting that $\mathcal{P}_{\Omega}(\mathcal{P}_{\Omega}-\mathcal{P}_{\Omega}\mathcal{P}_{\Phi}\mathcal{P}_{\Omega})^{-1}=(\mathcal{P}_{\Omega}-\mathcal{P}_{\Omega}\mathcal{P}_{\Phi}\mathcal{P}_{\Omega})^{-1}$. Therefore, by Hoeffding's inequality, conditional on $\Omega$, $U$ and $V$,
\begin{align*}
P\big(|\tilde{W}_{S,it}|> t\big|\Omega, U,V\big)\leq 2\exp(-\frac{2t^{2}}{||X(i,t)||_{F}^{2}})
\end{align*}
and by using the union bound,
\begin{align*}
P\big(||\tilde{W}_{S}||_{\infty}>t\big|\Omega, U,V\big)\leq 2NT\exp(-\frac{2t^{2}}{||X(i,t)||_{F}^{2}})
\end{align*}
Under the conditions in Lemma \ref{lem1} and the incoherence conditions,
\begin{equation*}
||\mathcal{P}_{\Omega}\mathcal{P}_{\Phi}(e_{i}e_{t}')||_{F}\leq ||\mathcal{P}_{\Omega}\mathcal{P}_{\Phi}||\cdot ||\mathcal{P}_{\Phi}(e_{i}e_{t}')||_{F}\leq \sqrt{1-\delta+\varepsilon \delta} \sqrt{C\mu r/N}
\end{equation*}
Meanwhile, from the proof of d), we know $||(\mathcal{P}_{\Omega}-\mathcal{P}_{\Omega}\mathcal{P}_{\Phi}\mathcal{P}_{\Omega})^{-1}||\leq \frac{1}{\delta(1-\varepsilon)}$. Therefore,
\begin{align*}
P\big(||W_{S}||_{\infty}>\lambda t|\Omega, U,V\big)\leq &P\big(||\tilde{W}_{S,it}||_{\infty}> t\big|\Omega, U,V\big)\\
\leq &2NT\exp\big(-\frac{C't^{2}N\delta^{2}(1-\varepsilon)^{2}}{\mu r(1-\delta+\varepsilon\delta)}\big)\\
\to& 0,\forall t>0
\end{align*}
Since the right-hand side does not depend on $\Omega$, $U$ and $V$, the inequality holds unconditionally, which completes the proof.
\nocite{*}

\bibliography{reference}
\end{document}